\documentclass[preprint,report]{elsarticle}

\usepackage{array,xspace,multirow,hhline,tikz,colortbl,tabularx,booktabs,fixltx2e,amsmath,amssymb,amsfonts,amsthm}
\usepackage{paralist}
\usepackage[boxed]{algorithm}
\usepackage[noend]{algorithmic}
\usepackage{amssymb,amsthm,xspace}

\usepackage{verbatim,ifthen}
\usepackage{enumitem}
\usepackage{url}
\usepackage{pifont}
\usepackage{ifthen}
\usepackage{calrsfs,mathrsfs}
\usepackage{bbding,pifont}
\usepackage{pgflibraryshapes}
\usetikzlibrary{arrows}
\usepackage{centernot}
\usetikzlibrary{decorations.pathreplacing}
\usetikzlibrary{patterns}
\usepackage{pgflibraryshapes}
\usetikzlibrary{positioning,chains,fit,shapes,calc}
	\usepackage{varioref}
	\usepackage{nicefrac}

	\usepackage{tabularx}
	\usepackage{mathtools}
	
	\newcommand*\circled[1]{\tikz[baseline=(char.base)]{
	            \node[shape=circle,draw,inner sep=2pt] (char) {#1};}}

	\renewcommand{\algorithmicrequire}{\textbf{Input:}}
	\renewcommand{\algorithmicensure}{\textbf{Output:}}
	\algsetup{linenodelimiter=\,}
	\algsetup{linenosize=\tiny}
	\algsetup{indent=2em}

\definecolor{light-gray}{gray}{0.9}

	\newcommand{\PS}{MRR\xspace}
	
\newcommand{\bor}[1][]{\ifthenelse{\equal{#1}{}}{\mathit{BOR}}{\mathit{BOR}(#1)}}

		\newcommand{\ml}[1][]{\ensuremath{\ifthenelse{\equal{#1}{}}{\mathit{ML}}{\mathit{ML}(#1)}}\xspace}
		\newcommand{\sml}[1][]{\ensuremath{\ifthenelse{\equal{#1}{}}{\mathit{SML}}{\mathit{SML}(#1)}}\xspace}
		\newcommand{\sd}[1][]{\ensuremath{\ifthenelse{\equal{#1}{}}{\mathit{SD}}{\mathit{SD}(#1)}}\xspace}
		\newcommand{\rsd}[1][]{\ensuremath{\ifthenelse{\equal{#1}{}}{\mathit{RSD}}{\mathit{RSD}(#1)}}\xspace}
		\newcommand{\rd}[1][]{\ensuremath{\ifthenelse{\equal{#1}{}}{\mathit{RD}}{\mathit{RD}(#1)}}\xspace}
		\newcommand{\st}[1][]{\ensuremath{\ifthenelse{\equal{#1}{}}{\mathit{ST}}{\mathit{ST}(#1)}}\xspace}
		\newcommand{\bd}[1][]{\ensuremath{\ifthenelse{\equal{#1}{}}{\mathit{BD}}{\mathit{BD}(#1)}}\xspace}
		\newcommand{\pc}[1][]{\ensuremath{\ifthenelse{\equal{#1}{}}{\mathit{PC}}{\mathit{PC}(#1)}}\xspace}
		\newcommand{\dl}[1][]{\ensuremath{\ifthenelse{\equal{#1}{}}{\mathit{DL}}{\mathit{DL}(#1)}}\xspace}
		\newcommand{\ul}[1][]{\ensuremath{\ifthenelse{\equal{#1}{}}{\mathit{UL}}{\mathit{UL}(#1)}}\xspace}

				\newcommand{\bt}[1][]{\ensuremath{\ifthenelse{\equal{#1}{}}{\mathit{BT}}{\mathit{BT}(#1)}}\xspace}

\newcommand{\egal}{\ensuremath{ESR}\xspace}
\newcommand{\pp}{\ensuremath{PP}\xspace}

			\newcommand{\pref}{\ensuremath{\succsim}}
			\newcommand{\spref}{\ensuremath{\succ}}
			\newcommand{\indiff}{\ensuremath{\sim}}

	\newtheorem{theorem}{Theorem}%
	\newtheorem{corollary}{Corollary}%
	\newtheorem{example}{Example}
	
	\newtheorem{remark}{Remark}
	
		\newtheorem{fact}{Fact}%
\usepackage{eqparbox}

	\newcommand{\ims}{IMS}

	\usepackage{enumitem}
	\setenumerate[1]{label=\rm(\it{\roman{*}}\rm),ref=({\it\roman{*}}),leftmargin=*}
	\newlength{\wordlength}
	\newcommand{\wordbox}[3][c]{\settowidth{\wordlength}{#3}\makebox[\wordlength][#1]{#2}}

	\newcommand{\midd}{\mathbin{:}}

\usepackage{enumitem}
\setenumerate[1]{label=\rm(\it{\roman{*}}\rm),ref=({\it\roman{*}}),leftmargin=*}

\newcommand{\nbh}[1][]{
	\ifthenelse{\equal{#1}{}}{\nu}{\nu(#1)}
}

\newcommand{\cstr}[1][]{
	\ifthenelse{\equal{#1}{}}{\mathscr S}{\cstr(#1)}
}

\newcommand{\choice}[1][]{
	\ifthenelse{\equal{#1}{}}{\mathit{C}}{\choice(#1)}
}

\newcommand{\guar}{\ell}

\begin{document}

	 \title{Participation Incentives in Randomized Social Choice}

	

	\author{Haris Aziz\corref{cor1}} \ead{haris.aziz@data61.csiro.au}
	%
	\address{Data61 and UNSW,\\ Sydney, Australia}

\begin{abstract}
When aggregating preferences of agents via voting, two desirable goals are to identify outcomes that are Pareto optimal and to incentivize agents to participate  in the voting process. We consider participation notions as formalized by Brandl, Brandt, and Hofbauer (2015) and study how far efficiency and participation are achievable by randomized social choice functions in particular when agents' preferences are downward lexicographic (\dl) or satisfy stochastic dominance (\sd). Our results include the followings ones: we prove formal relations between the participation notions with respect to \sd and \dl and we show that the maximal recursive rule satisfies very strong participation with respect to both \sd and \dl. 
\end{abstract}

\maketitle
	
            		\section{Introduction}

            	Two fundamental goals in collective decision making are (1) agents should be incentivized to participate and (2) the outcome should be such that there exists no other outcome that each agent prefers. 	
            	We consider these goals of \emph{participation}~\citep{BrFi83a,Moul88b} and \emph{efficiency}~\citep{ABB14b,Moul03a} in the context of randomized social choice. In randomized social choice, we study randomized social choice functions (referred to as \emph{social decision schemes (SDSs)} which take as input agents' preferences over alternatives and return a probability distribution over the alternatives. The probability distribution can also represent time-sharing arrangements or relative importance of alternatives~\citep{Aziz13b,BMS05a}. For example, agents may vote on the proportion of time different genres of songs are played on a radio channel. This type of preference aggregation is not captured by traditional deterministic voting in which the output is a single discrete alternative which may not be suitable to cater for different tastes.

            	When defining notions such as participation, efficiency, and strategyproofness, one needs to reason about preferences over probability distributions (lotteries). In order to define these properties, we consider the \emph{stochastic dominance} (\sd) and \emph{downward lexicographic} (\dl) lottery extensions to extend preferences over alternatives to preferences over lotteries. A lottery is preferred over another lottery with respect to \sd, if for all utility functions consistent with the ordinal preferences, the former yields as much utility as the latter. \dl is a natural lexicographical refinement of \sd. Lexicographic preferences have received considerable attention within randomized social choice~\cite{AbSo03a, Aziz13b,AzSt14a,Cho12a,SaSe13b,ScVa12a}.

            	Although efficiency and strategyproofness with respect to \sd and \dl have been considered in a series of papers~\citep{Aziz13b,AzSt14a,ABBH12a,ABB13d,BMS05a,Cho12a,Gibb77a,Proc10a}, three notions of participation with respect to \sd  were formalized only recently by \citet{BBH15b}.
            The three notions include very strong (participating is strictly beneficial), strong (participating is at least as helpful as not participating) and standard (not participating is not more beneficial). In contrast to deterministic social choice in which the number of possible outcome is at most the number of alternatives, randomized social choice admits infinite outcomes which makes participation even more meaningful: an agent may be able to perturb the outcome of the lottery slightly in his favour by participating in the voting process. In spirit of the radio channel example, voters should ideally be able to increase the fractional time of their favorite music genres by participating in the vote to decide the durations.  

            Participation is closely related to strategyproofness which requires that misreporting preferences is not beneficial. If agents are truthful  but may consider not participating in voting, then the issue of participation assumes more importance than the issue of untruthful voting~\citep{BBH15b}. Note that for almost all reasonable social choice functions, participating but expressing complete indifference between all alternatives is equivalent to not participating at all.
            	The two central results presented by \citet{BBH15b} were:
            	 \begin{inparaenum}[(1)]
            		\item there exists a social decision scheme (RSD) that satisfies very strong \sd-participation and ex post efficiency (Theorem 4, \citep{BBH15b});
            		\item There exists a social decision scheme (uniform randomization over the Borda winners) that satisfies strong \sd-participation and \sd-efficiency (Theorem 7, \citep{BBH15b}).
            	 \end{inparaenum}


            	Using the work by \citet{BBH15b} as a starting point, we expand the discussion on participation in randomized social choice by considering participation with respect to the lexicographic lottery extension and exploring the relationship between the participation notions with respect to the two extensions. We also consider social decision schemes that were not considered by \citet{BBH15b}) including the 
            \emph{maximal recursive rule (MR)}~\citep{Aziz13b};\emph{egalitarian simultaneous reservation (ESR)}~\citep{AzSt14a};
	  and \emph{serial dictatorship}.
            	We consider the extent to which participation can be achieved by SDSs.

            			\begin{table*}[h!]
            				\centering
            			\label{tab:compare}
            			\begin{tabular}{lccccc}
            			\toprule
            			&\rsd&$\sml$&$BO$&$MR$&\egal\\
            			Properties&&\\
            				\midrule
            			\dl-efficient&--&--&--&--&+\\
            			\sd-efficient&--&+&+&--&+\\
            			ex post efficient&+&+&+&+&+\\
            			\midrule
            			Very strong \sd-participation&+&--&--&\circled{+}&\circled{--}\\
            			Very strong \dl-participation&+&--&\circled{--}&\circled{+}&\circled{--}\\
            			Strong \sd-participation&+&--&+&\circled{+}&\circled{--}\\
            			\bottomrule
            			\end{tabular}
            			\caption{A comparison of axiomatic properties of different social decision schemes: $RSD$ (random serial dictatorship), $SML$ (strict maximal lotteries), $BO$ (uniform randomization over Borda winners), $MR$ (maximal recursive rule), and \egal (egalitarian simultaneous reservation). The circled results are from this paper. All the schemes are anonymous and neutral.}
            			\label{table:summary:egal}
            			\end{table*}

            	\paragraph{Contributions}
	
            	Our contributions include relations among participation notions (see Figure~\ref{fig:part-relations}) as well as understanding the relative merits of SDSs in terms of efficiency and participation (see Table~\ref{table:summary:egal}). 

            	\begin{itemize}
            		\item We relate participation concepts with respect to \dl and \sd (see Figure~\ref{fig:part-relations}). Although very strong \sd-participation implies very strong \dl-participation, \dl-participation implies \sd-participation. Moreover a combination of strong \sd-participation and \dl-participation implies very strong \sd-participation.
            		\item We show that 	the Maximal Recursive ($MR$) rule~\citep{Aziz13b} satisfies very strong \sd-participation hence being the first known SDS to date to satisfy the property and also be  polynomial-time computable. Previously, \rsd was proved to satisfy very strong \sd-participation~\citep{BBH15b} but 
            \rsd probabilities are \#P-complete to compute~\citep{ABB13b}.
		
	
            	\item We point that \egal~\citep{AzSt14a} does not satisfy strong \sd-participation.
	
            			\item We highlight that although random serial dictatorship satisfies very strong \sd-participation and hence very strong \dl-participation, serial dictatorship does not even satisfy very strong \dl-participation. 
            It follows that if a rule satisfies $\dl$-strategyproofness, it need not satisfy very strong \dl-participation. Similarly, if a rule satisfies $\sd$-strategyproofness, it need not satisfy very strong \sd-participation.
            		\end{itemize}

            \section{Related Work}

            One of the first formal works on randomized social choice is by \citet{Gibb77a}. 
            The literature in randomized social choice has grown over the years 
            although it is much less developed in comparison to deterministic social choice. The main result of \citet{Gibb77a} was that random dictatorship in which each agent has uniform probability of choosing his most preferred alternative is the unique anonymous, strategyproof and ex post efficient SDS. Random serial dictatorship (\rsd) is the natural generalization of random dictatorship for weak preferences but the \rsd lottery is \#P-complete to compute~\citep{ABB13b}. 

            \citet{BoMo01a} initiated the use of stochastic dominance to consider various notions of strategyproofness, efficiency, and fairness conditions in the domain of \emph{random assignments} which is a special type of social choice setting. They proposed the 
            probabilistic serial mechanism --- a desirable random assignment mechanism. 
            \citet{Cho12a} extended the approach of \citet{BoMo01a} by considering other lottery extensions such as ones based on lexicographic preferences.

            The tradeoff of efficiency and strategyproofness for SDSs was formally considered in a series of papers~\citep{Aziz13b, AzSt14a, ABBH12a, ABB13d, BMS05a}.
            \citet{AzSt14a} presented a generalization --- \emph{Egalitarian Simultaneous Reservation} (\egal) --- of the probabilistic serial mechanism to the domain of social choice.  \citet{Aziz13b} proposed the \emph{maximal recursive (\PS)} SDS which is similar to the random serial dictatorship but for which the lottery can be computed in polynomial time.  

            \citet{BBH15b} showed that the strict maximal lottery SDS satisfies \sd-efficiency and \sd-participation; uniform randomization over Borda winners satisfies strong \sd-participation; and \rsd (random serial dictatorship) satisfies very strong \sd-participation. The main open problem posed by \citet{BBH15b} was whether there exists an SDS that satisfies very strong SD-participation and SD-efficiency. Although random dictatorship (defined for strict preferences) satisfies both properties, it is unclear whether an SDS satisfies both properties when agents may express ties in their preferences. In more recent work, \citet{BBH15d} study the connection between welfare maximization and participation.


            		\section{Preliminaries}

            			Consider the social choice setting in which there is a set of agents $N=\{1,\ldots, n\}$, a set of alternatives $A=\{a_1,\ldots, a_m\}$ and a preference profile $\pref=(\pref_1,\ldots,\pref_n)$ such that each $\pref_i$ is a complete and transitive relation over $A$. Let $\mathcal{F}(\mathbb{N})$ denote the set of all finite and non-empty subsets of $\mathbb{N}$.
            			We write~$a \pref_i b$ to denote that agent~$i$ values alternative~$a$ at least as much as alternative~$b$ and use~$\spref_i$ for the strict part of~$\pref_i$, i.e.,~$a \spref_i b$ iff~$a \pref_i b$ but not~$b \pref_i a$. Finally, $\indiff_i$ denotes~$i$'s indifference relation, i.e., $a \indiff_i b$ iff both~$a \pref_i b$ and~$b \pref_i a$.
            			The relation $\pref_i$ results in equivalence classes $E_i^1,E_i^2, \ldots, E_i^{k_i}$ for some $k_i$ such that $a\spref_i a'$ iff $a\in E_i^l$ and $a'\in E_i^{l'}$ for some $l<l'$. Often, we will use these equivalence classes to represent the preference relation of an agent as a preference list
            			$i\midd E_i^1,E_i^2, \ldots, E_i^{k_i}$.
            		For example, we will denote the preferences $a\indiff_i b\spref_i c$ by the list $i:\ \{a,b\}, \{c\}$. 
            		For any set of alternative $A'$, we will refer by $\max_{\pref_i}(A')$ the set of most preferred alternatives according to preference $\pref_i$.

            		An agent $i$'s preferences are \emph{dichotomous} iff he partitions the alternatives into just two equivalence classes, i.e., $k_i=2$. An agent $i$'s preferences are \emph{strict} iff $\pref_i$ is antisymmetric, i.e.
            		 all equivalence classes have size 1.

            			Let $\Delta(A)$ denote the set of all \emph{lotteries} (or \emph{probability distributions}) over $A$.
            			The support of a lottery $p \in \Delta(A)$, denoted by $\text{supp}(p)$, is the set of all alternatives to which $p$ assigns a positive probability, i.e., $\text{supp}(p) = \{x \in A \mid p(x)>0\}$. We will write $p(a)$ for the probability of alternative $a$ and we will represent a lottery as
            			$p_1a_1 + \cdots + p_ma_{m}$
            			where $p_j=p(a_j)$ for $j\in \{1,\ldots, m\}$. For $A'\subseteq A$, we will (slightly abusing notation) denote $\sum_{a\in A'}p(a)$ by $p(A')$.

            			A \emph{social decision scheme} is a function $f: \mathcal{R}^n \rightarrow \Delta(A)$. If $f$ yields a set rather than a single lottery, we call $f$ a \emph{correspondence}.
            		Two minimal fairness conditions for SDSs are \emph{anonymity} and \emph{neutrality}. Informally, they require that the SDS should not depend on the names of the agents or alternatives respectively.

            			In order to reason about the outcomes of SDSs, we need to determine how agents compare lotteries. A \emph{lottery extension} extends preferences over alternatives to (possibly incomplete) preferences over lotteries.
            			Given $\pref_i$ over $A$, a \emph{lottery extension} $\mathcal{E}$ extends $\pref_i$ to $\pref_i^\mathcal{E}$ over the set of lotteries $\Delta(A)$. We now define some particular lottery extensions that we will later refer to.
            			\begin{itemize}
            		\item 	Under \emph{stochastic dominance (SD)}, an agent prefers a lottery that, for each alternative $x \in A$, has a higher probability of selecting an alternative that is at least as good as $x$. Formally, $p \pref_i^{\sd} q$ iff $\forall y\in\nolinebreak A \colon \sum_{x \in A: x \pref_i y} p(x) \geq \sum_{x \in A: x \pref_i y} q(x).$

            		\item 	In the  \emph{downward lexicographic (\dl)} extension, an agent prefers the lottery with higher probability for his most preferred equivalence class, in case of equality, the one with higher probability for the second most preferred equivalence class, and so on.
            			Formally, $p \pref_i^{\dl} q$ iff for the smallest (if any) $l$ with $p(E^l_i) \neq q(E^l_i)$ we have $p(E^l_i) > q(E^l_i)$.
            		\end{itemize}

            		\sd~\cite{BoMo01a} is particularly important because $p \pref^{\sd} q$ iff $p$ yields at least as much expected utility as $q$ for any von-Neumann-Morgenstern utility function consistent with the ordinal preferences \cite{Cho12a}.

            		We say a lottery extension $\mathcal{E}$ is a \emph{refinement} of $\mathcal{E}'$ if 
            		\[p\pref_i^{\mathcal{E}'} q \iff p \succ_i^{\mathcal{E}} q.\]
		
            		We say a lottery extension $\mathcal{E}$  is \emph{complete} if $p \pref_i^{\mathcal{E}} q$ or $q \pref_i^{\mathcal{E}} p$ for all $p,q\in \Delta(A)$ and $\pref_i$.
            \dl refines \sd to a complete relation based on the natural lexicographic relation over lotteries~\cite{ScVa12a,AbSo03a,Cho12a}.

            The following example illustrates a social choice setting where randomized outcomes are compared by an agent with respect to the \sd and \dl relations.
            \begin{example}
            	Consider the preference profile:
            	\begin{align*}
            		1:\quad&a,b,c,d\\
            		2:\quad&\{a,b\}, \{c,d\}\\
            		3:\quad&\{c,d\},\{a,b\}
            	\end{align*}
            Agent 1 most prefers $a$, then $b$, $c$, and $d$ whereas agent $2$ is indifferent between $a$ and $b$.
            Then $\frac{1}{2}a+\frac{1}{2}c$ is a possible randomized outcome in which the probability of $a$ and $c$ is half each.
            Note that that $\frac{2}{3}a+\frac{1}{3}d \not\pref_1^{\sd} \frac{1}{2}a+\frac{1}{2}c$ but $\frac{2}{3}a+\frac{1}{3}d \not\pref_i^{\dl} \frac{1}{2}a+\frac{1}{2}c$.
            	\end{example}

            		\paragraph{Efficiency and strategyproofness}

            		Let ${\mathcal{E}}$ be any lottery extension. A lottery $p$ is \emph{${\mathcal{E}}$-efficient} iff there exists no lottery $q$ such that $q \pref_i^{\mathcal{E}} p$ for all $i\in N$ and $q \spref_i^{\mathcal{E}} p$ for some $i\in N$. An SDS is ${\mathcal{E}}$-efficient iff it always returns an ${\mathcal{E}}$-efficient lottery. A standard efficiency notion that cannot be phrased in terms of lottery extensions is \emph{ex post efficiency}. A lottery is ex post efficient iff it is a lottery over Pareto optimal alternatives. It is the case that 	\dl-efficiency $\implies$ \sd-efficiency $\implies$ ex post efficiency. 

            		\begin{figure}[h!]

            			\begin{center}

            			\scalebox{0.9}{
            			\begin{tikzpicture}
            			\tikzstyle{pfeil}=[->,>=angle 60, shorten >=1pt,draw]
            			\tikzstyle{onlytext}=[]

            				\node[onlytext] (DL) at (0,2) {$\dl$-efficiency};
            			\node[onlytext] (SD) at (0,0) { $\sd$-efficiency};
            			\node[onlytext] (expost) at (0,-2) {ex post efficiency};

            			 \draw[pfeil] (DL) to (SD);
            			 \draw[pfeil] (SD) to (expost);

            			\end{tikzpicture}

            			}

            			\end{center}	\caption{Relations between efficiency concepts.}
            			\label{fig:part-relations}
            		\end{figure}
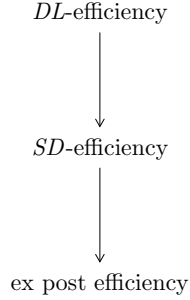

            		An SDS~$f$ is ${\mathcal{E}}$-\emph{ma\-nip\-u\-la\-ble} iff there exists an agent $i \in N$ and preference profiles $\pref$ and $\pref'$ with $\pref_j=\pref_j'$ for all $j \ne i$ such that $f(\pref') \spref_i^{\mathcal{E}} f(\pref)$.
            		An SDS is \emph{weakly} ${\mathcal{E}}$-\emph{strategyproof} iff it is not ${\mathcal{E}}$-manipulable,
            		it is ${\mathcal{E}}$-\emph{strategyproof} iff $f(\pref) \pref_i^{\mathcal{E}} f(\pref')$ for all $\pref$ and $\pref'$ with $\pref_j=\pref_j'$ for all $j \neq i$. Note that \sd-strategyproofness is equivalent to strategyproofness in the Gibbard sense. It is known that \sd-strategyproof  $\implies$ \dl-strategyproof $\implies$ weak \sd-strategyproof.


            	\section{Participation}

            	For any lottery extension $\mathcal{E}$, we can define three notions of participation~\citep{BBH15b}.
	
            	\begin{figure*}[h!]

            		\begin{center}

            		\scalebox{0.7}{
            		\begin{tikzpicture}
            		\tikzstyle{pfeil}=[->,>=angle 60, shorten >=1pt,draw]
            		\tikzstyle{onlytext}=[]


            		\node[onlytext] (vsSD) at (4,4) {\large very strong $\sd$-participation};
            			\node[onlytext] (vsDL) at (8,2) {\large very strong $\dl$-participation};
            			\node[onlytext] (sDL) at (8,0) {\large strong $\dl$-participation};
            			\node[onlytext] (DL) at (8,-2) {\large $\dl$-participation};
            		\node[onlytext] (sSD) at (0,2) {\large strong $\sd$-participation};
            		\node[onlytext] (SD) at (0,-4) {\large $\sd$-participation};

            		 \draw[pfeil] (vsSD) to (vsDL);
            		  \draw[pfeil] (vsSD) to (sSD);
            		   \draw[pfeil] (vsDL) to (sDL);
            		    \draw[pfeil] (sSD) to (SD);
            	 \draw[pfeil] (sSD) to (sDL);
            	 	 \draw[pfeil] (sDL) to (DL);
            		  	 \draw[pfeil] (DL) to (sDL);
            			 	  	 \draw[pfeil] (DL) to (SD);

            		\end{tikzpicture}

            		}

            		\end{center}	\caption{Relations between participation concepts.}
            		\label{fig:part-relations}
            	\end{figure*}
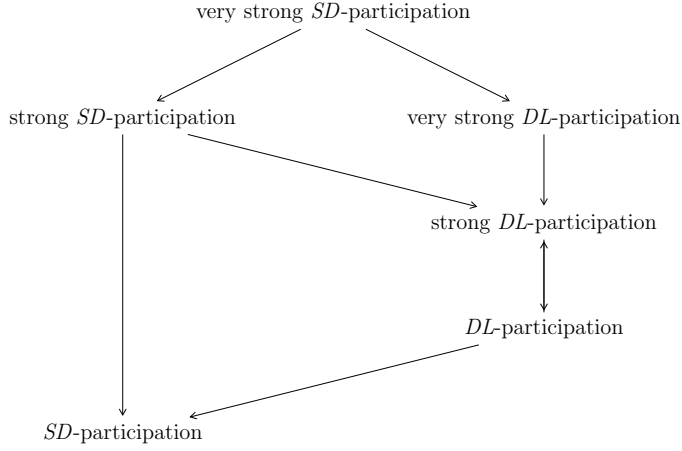

            	%
            	%
            	%
            	%
            	%
            	%
            	%
            	%
            	%

            	\begin{itemize}
            		\item Formally, an SDS $f$ is $\mathcal{E}$-manipulable (by strategic abstention) if there exist $\pref \in\mathcal{\pref}^N$ for some $N\in\mathcal{F}(\mathbb{N})$ and $i\in N$ such that $f(\pref_{-i}) \mathrel{\succ_i^{\mathcal{E}}} f(\pref)$. If an SDS is not $\mathcal{E}$-manipulable it satisfies \emph{$\mathcal{E}$-participation}.
            		\item An SDS $f$ satisfies \emph{strong $\mathcal{E}$-participation} if $f(\pref) \mathrel{\pref_i^{\mathcal{E}}} f(\pref_{-i})$ for all $N\in\mathcal{F}(\mathbb{N})$, $\pref\in\mathcal{R}^N$, and $i\in N$. 
            		\item An SDS $f$ satisfies \emph{very strong $\mathcal{E}$-participation} if  for all $N\in\mathcal{F}(\mathbb{N})$, $\pref\in\mathcal{R}^N$, and $i \in N$, $f(\pref)\mathrel{\pref_i^{\mathcal{E}}} f(\pref_{-i})$ and 
            	\[
            	  f(\pref)\mathrel{\succ_i^{\mathcal{E}}} f(\pref_{-i}) \text{ whenever }\exists p\in \Delta(A)\colon p\mathrel{\succ_i^{\mathcal{E}}} f(\pref_{-i}).
            	\]
            	\end{itemize}

            	\begin{fact}
            		For any lottery extension, very strong $\mathcal{E}$-participation implies strong $\mathcal{E}$-participation which implies $\mathcal{E}$-participation.
            		\end{fact}

            		Next, we make further general observations about the relation between participation notions.
	
            		\begin{theorem}
            		For any complete lottery extension, strong $\mathcal{E}$-participation is equivalent to $\mathcal{E}$-participation.
            			\end{theorem}
            			\begin{proof}
            				If an SDS satisfies $\mathcal{E}$-participation, then 
            				it cannot be that $f(\pref_{-i}) \mathrel{\succ_i^{\mathcal{E}}} f(\pref)$. Since $\mathcal{E}$ is complete, the statement is equivalent to saying that $f(\pref) \mathrel{\pref_i^{\mathcal{E}}} f(\pref_{-i})$ which is equivalent to satisfying strong $\mathcal{E}$-participation.
            				\end{proof}
	
            		\begin{theorem}
            		For any complete lottery extension $\mathcal{E}$ that is a refinement of \sd, the following relations hold: 
            		\begin{enumerate}
            			\item 	Strong \sd-participation implies Strong $\mathcal{E}$-participation.
            			\item 	Very strong \sd-participation implies very strong $\mathcal{E}$-participation.
            			\item $\mathcal{E}$-participation implies \sd-participation.
            		\end{enumerate}
            			\end{theorem}
            \begin{proof}
            	Consider a complete lottery extension $\mathcal{E}$ that is a refinement of \sd
            	\begin{enumerate}
            		\item An SDS $f$ satisfies strong \sd-participation if $f(\pref) \mathrel{\pref_i^{\sd}} f(\pref_{-i})$ which implies that $f(\pref) \mathrel{\pref_i^{\mathcal{E}}} f(\pref_{-i})$ which is equivalent to $f$ satisfying strong $\mathcal{E}$-participation.
            		\item If an SDS $f$ satisfies very strong \sd-participation, it satisfies  
            		strong \sd-participation which means it satisfies strong $\mathcal{E}$-participation. Since $f$ satisfies very strong \sd-participation, 
            $f(\pref)\mathrel{\succ_i^{\sd}} f(\pref_{-i}) \text{ whenever }\exists p\in \Delta(A)\colon p\mathrel{\succ_i^{\mathcal{E}}} f(\pref_{-i}).$ Since $\mathcal{E}$ is a refinement of \sd, it implies that $f(\pref)\mathrel{\succ_i^{\mathcal{E}}} f(\pref_{-i}) \text{ whenever }\exists p\in \Delta(A)\colon p\mathrel{\succ_i^{\mathcal{E}}} f(\pref_{-i}).$
            \item Assume $f$ does not satisfy \sd-participation. Then $f(\pref_{-i}) \mathrel{\succ_i^{\sd}} f(\pref)$ for some profile $\pref$. This implies that $f(\pref_{-i}) \mathrel{\succ_i^{\mathcal{E}}} f(\pref)$ which means that $f$ does not satisfy $\mathcal{E}$-participation.
            	\end{enumerate}
            	\end{proof}

            	\begin{corollary}
            		The following relations hold:
            		\begin{enumerate}
            			\item 	Strong \sd-participation implies Strong \dl-participation.
            			\item 	Very strong SD-participation implies very strong DL-participation.
            			\item DL-participation implies \sd-participation.
            		\end{enumerate}
            		\end{corollary}

            The following statement also follows directly from the definitions.

            		\begin{theorem}\label{th:comb}
            		For any complete lottery extension $\mathcal{E}$ that is a refinement of \sd, 
            the combination of strong \sd-participation and  very strong $\mathcal{E}$-participation implies very strong \sd-participation.
            			\end{theorem}
            \begin{proof}
            	An SDS $f$ satisfies \emph{very strong $\mathcal{E}$-participation} if  for all $N\in\mathcal{F}(\mathbb{N})$, $\pref\in\mathcal{R}^N$, and $i \in N$, $f(\pref)\mathrel{\pref_i^{\mathcal{E}}} f(\pref_{-i})$ and 
            		\[
            		  f(\pref)\mathrel{\succ_i^{\mathcal{E}}} f(\pref_{-i}) \text{ whenever }\exists p\in \Delta(A)\colon p\mathrel{\succ_i^{\mathcal{E}}} f(\pref_{-i}).
            		\]
		
            		Assume that $f(\pref_{-i})$ is such that $\exists p\in \Delta(A)\colon p\mathrel{\succ_i^{\mathcal{E}}} f(\pref_{-i})$
            		Now assume that $f$ satisfies very strong $\mathcal{E}$-participation and strong \sd-participation but not very strong \sd-participation. But this means that $f(\pref)\mathrel{\succ_i^{\mathcal{E}}} f(\pref_{-i})$ and $f(\pref)\mathrel{\succsim_i^{\mathcal{\sd}}} f(\pref_{-i})$. 
            Since 	$f(\pref)\mathrel{\succ_i^{\mathcal{E}}} f(\pref_{-i})$, we know that 	$f(\pref)\mathrel{\not\sim_i^{\mathcal{\sd}}} f(\pref_{-i})$. Hence 
             $f(\pref)\mathrel{\succ_i^{\mathcal{\sd}}} f(\pref_{-i})$ which means that $f$ satisfies very strong $\sd$-participation.
            	\end{proof}




            	\section{Social Decision Schemes}
	
            	In this paper, we formally prove participation properties of classic as well as recently introduced SDSs. Before we do that, we give an overview of the SDSs.
	
            	\subsection{Serial Dictatorship and Random Serial Dictatorship}
	
            	 The \emph{serial dictatorship} rule is defined with respect to a permutation $\pi$ over $N$. It starts with the set of all alternatives and then each agent in $\pi$ successively refines the set of alternatives to the set of most preferred alternatives from the remaining set. \rsd returns the serial dictatorship outcome with respect to a permutation that is chosen uniformly at random.
	 
            	 		\subsection{Egalitarian Simultaneous Reservation}


            	 The egalitarian simultaneous reservation (\egal) rule is based on the idea of gradually refining the set of lotteries. We present verbatim the informal description of \egal as presented by \citet{AzSt14a}:

            	 Starting from the entire set $\Delta(A)$, the \egal algorithm proceeds by gradually restricting the set of possible outcomes. The restrictions enforced are lower bounds for the probability of certain equivalence classes while it is always maintained that a lottery exists that satisfies all these lower bounds.
            	 	Each equivalence class $E$ is represented by a \emph{tower} where at any time $t$, the height of this tower's \emph{ceiling} $\guar_t(E)$ represents the lower bound in place for the probability of this subset at that time.
            	 	During the algorithm, agents will climb up these towers and try to push up the ceilings, thereby increasing the lower bounds for certain subsets. 
            	 Each tower starts with the height of its ceiling set to 0. Every agent starts climbing up the tower that corresponds to his most preferred equivalence class. Whenever an agent hits the ceiling, he tries to push it up. He continues climbing, pushing up the ceiling at the same time. Note that the ceiling will only be pushed up as fast as the agent pushing it can climb. Two agents pushing up a ceiling at the same time therefore does not increase the speed of it being pushed up. When it cannot be pushed up any further without compromising the existence of a lottery satisfying all current lower bounds, we say that set $E$ is \emph{tight} and has been \emph{frozen}. 
            	 	At this point, the agent \emph{bounces off the ceiling} and drops back to the floor, moving on to the tower corresponding to his next most preferred equivalence class. 
            	 	 We can think of the algorithm proceeding in stages where a stage ends whenever some agent bounces off the ceiling.
            	 	The algorithm ends when all the equivalence classes have been frozen at which point some lottery satisfying the lower bounds is returned. 


            	\subsection{Maximal Recursive Rule}

	 We now describe the MR (Maximal Recursive) rule as presented by \citet{Aziz13b}.
	 We will denote by	$s^1(a,S,\pref)$ the \emph{generalized plurality score} of $a$ according to $\pref$ when the alternative set and the preference profile is restricted to $S$.  
	 		$$s^1(a,S,\pref)=|\{i\in N\midd a\in \max_{\pref_i}(S)\}|.$$

	 		 \PS relies on the concept of \ims ~(inclusion minimal subsets).
	 		For $S\subseteq A$, let $A_1,\ldots, A_{m'}$ be subsets of $S$.
	 		Let $I(A_1,\ldots, A_{m'})$ be the set of non-empty intersections of elements of some subset of $\{A_1,\ldots, A_{m'}\}$.
	 		\begin{align*}
	 			&I(A_1,\ldots, A_{m'})\\
	 			=&\{X\in 2^A\setminus \emptyset \midd X= \bigcap_{A_j\in A'}A_j
	 			 \textrm{~for some~} A'\subseteq \{A_1,\ldots, A_{m'}\}\}.
	 		\end{align*}


	 		Then, the inclusion minimal subsets of $S\subseteq A$ with respect to $(A_1,\ldots, A_{m'})$ are defined as follows.
	 		\begin{align*}
	 			&\ims(A_1,\ldots, A_{m'})\\
	 			=&\{X\in I(A_1,\ldots, A_{m'})\midd \nexists X'\in I(A_1,\ldots, A_{m'})\\
	 		& \textrm{~s.t.~} X'\subset X\}.
	 		\end{align*}

	 			\PS is defined as Algorithm~\ref{algo:main} that requires as a subroutine Algorithm~\ref{algo:subroutine} that involves calls on subsets of the alternatives. 

				\begin{algorithm}[h!]
				\caption{\PS}
				\label{PS}
				\centering
				\renewcommand{\algorithmicrequire}{\wordbox[l]{\textbf{Input}:}{\textbf{Output}:}} 
				\renewcommand{\algorithmicensure}{\wordbox[l]{\textbf{Output}:}{\textbf{Output}:}}
				\begin{algorithmic}

				\REQUIRE $(A,N,\pref)$
				\end{algorithmic}
				\algsetup{linenodelimiter=\,}
				\begin{algorithmic} 
				\STATE Call $\text{\PS-subroutine}(A,1,(A,N,\pref))$ to compute $p(a)$ for each $a\in A$.
				\RETURN $[a_1:p(a_1),\ldots, a_m: p(a_m)]$
				\end{algorithmic}
				\label{algo:main}
				\end{algorithm}

 			\begin{algorithm}[t!]
 			  \caption{\PS-subroutine}
 			  \label{PS}
 			\small
 			\renewcommand{\algorithmicrequire}{\wordbox[l]{\textbf{Input}:}{\textbf{Output}:}} 
 			 \renewcommand{\algorithmicensure}{\wordbox[l]{\textbf{Output}:}{\textbf{Output}:}}
 			\begin{algorithmic}
				\small
 				\REQUIRE $(S,v,(A,N,\pref))$
 			\end{algorithmic}
 			\algsetup{linenodelimiter=\,}
 			  \begin{algorithmic}[1] 
				  \small

 				\IF{$\max_{\pref_i}(S)=S$ for all $i\in N$}
 				\STATE $p(a)=v/|S|$ for all $a\in S$
 				\ELSE

 		\STATE $T(i,S,\pref)\longleftarrow\{a\midd a\in \arg \max_{a\in \max_{\pref_i}(S)} s^1(a,S,\pref)\}$ for all $i\in N$\\ 
 			\STATE $t(i,a,\pref)\longleftarrow1/|T(i,S,\pref)|$ if $a\in T(i,S,\pref)$ \& zero otherwise for all $i\in N$ \& $a\in S$  
 		\STATE $\gamma(a)\longleftarrow\sum_{i\in N}t(i,a,\pref)$ for all $a\in S$
 		\STATE $p(a)\longleftarrow v(\gamma(a,\pref))/|N|$ for all $a\in S$
 		\STATE $\{S_1,\ldots, S_k\}\longleftarrow \ims(\max_{\pref_1}(S),\ldots, \max_{\pref_n}(S))$
 		\FOR{each $S_j\in \{S_1,\ldots, S_k\}$}
 		\RETURN $\text{\PS-subroutine}(S_j,p(S_j), (A, N, \pref))$ 
 		\ENDFOR
 		\ENDIF  
 			 \end{algorithmic}
 			\label{algo:subroutine}
 			\end{algorithm}

							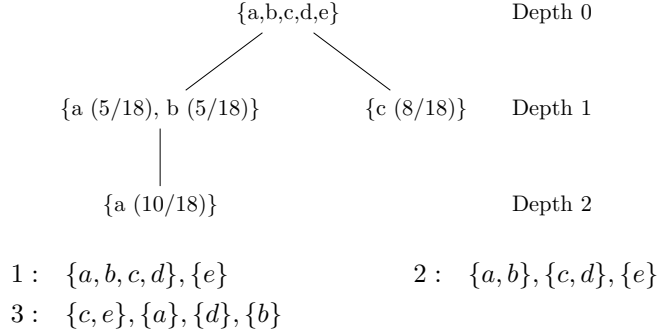
\begin{figure}
									\scalebox{0.85}{
								\centering
							\begin{tikzpicture}[level/.style={sibling distance=40mm/#1}]


							    \node { \{a,b,c,d,e\}}
							        child { node {\{a (5/18), b (5/18)\}} 
							          child {node {\{a (10/18)\}}}
							}
							        child {
							            node {\{c (8/18)\}}
							        };

						\draw (5,0) node () [label=left:$\text{Depth~} 0$] {};
							        \draw (5,-1.5) node () [label=left:$\text{Depth~} 1$] {};
							  \draw (5,-3) node () [label=left:$\text{Depth~} 2$] {};

							\end{tikzpicture}
						}
						\centering
							\begin{align*}
								1:\quad&\{a,b,c,d\}, \{e\}&	2:\quad&\{a,b\}, \{c,d\}, \{e\}\\
								3:\quad&\{c,e\},\{a\},\{d\},\{b\}
							\end{align*}
						\normalsize
							\caption{Recursion tree corresponding to running \PS on the preference profile above. The lottery returned is $[a: 10/18, b: 0, c: 8/18]$ when all agents vote. The outcome is $[c: 1]$ when $2$ does not vote. Hence agent $2$ gets an \sd more preferred outcome when he participates.
							}
		\label{fig:mrr}
							\end{figure}

            		\section{Participation Incentives in Social Decision Schemes}

            	Note that \rsd satisfies very strong \sd-participation~\citep{BBH15b} which implies both strong \sd-participation as well as very strong \dl-participation. However \rsd takes exponential time. We first observe that there is a simple rule that satisfies strong \sd-participation.

            	\begin{remark}
            	 The constant rule that gives the same probability to each alternative satisfies strong \sd-participation.
            	\end{remark}
	
            Since the constant rule is highly inefficient for any reasonable notion of efficiency, the remark shows that satisfying strong \sd-participation only becomes challenging when the goal is to additionally satisfy properties like ex post efficiency or \sd-efficiency.

            We first observe that $BO$ does not satisfy very strong $\dl$-participation. The reason is that an agent voting may not change the set of Borda winners which means that the resultant lottery does not change as well. 

            	%
            	%

            	Next, we show that there is a simple linear-time rule that achieves very strong \dl-strategyproofness. 
            	Consider the following rule which we refer to as \pp (Proportional Plurality).
            	Each agent has a total of 1 point which he uniformly distribute among the alternatives in his first equivalence class. Each alternative then gets probability that is equal the total amount of points that it gets divided by $n$. 

            	\begin{theorem}
            		\pp satisfies very strong \dl-participation.
            		\end{theorem}
            		\begin{proof}
            	Let us compare $\pp(\pref_{-i})$ with $\pp(\pref)$. In $\pp(\pref)$, the points of all alternative not in $\max_{i}(A)$ stay the same whereas the points of alternatives in $\max_{i}(A)$ increase because of the extra points allocated due to the presence of agent $i$. Hence the probability weight of agent $i$'s first equivalence class is strictly more in $\pp(\pref)$ in contrast to $\pp(\pref_{-i})$.		
            						\end{proof}
				
            In \pp, we have a natural rule that is anonymous, neutral and satisfies very strong \dl-participation.
            					Note however that \pp is not ex post efficient.
            					This raises the question whether there is any SDS that satisfies ex post efficiency as well as \dl-participation. 
            					 Next, we show that the \emph{Maximal Recursive (MR)} that is known to be ex post efficient~\citep{Aziz13b} also satisfies very strong \dl-participation.	

							 		Note that \rsd satisfies very strong \sd-participation~\citep{BBH15b} which implies \sd-participation. 
							 		However, it does not satisfy \sd-efficiency~\citep{ABBH12a} and it takes exponential time. 

							 	\begin{theorem}\label{th:mr-strong-sd}
							 									$MR$ satisfies very strong \sd-participation.
							 									\end{theorem}

							 									\begin{proof}

							 We first prove that by participating, an agent increases the probability of his first equivalence class if it is not already one. Let us compare $MR(\pref_{-i})$ with $MR(\pref)$. In $MR(\pref)$, the points of all alternative not in $\max_{i}(A)$ stay the same whereas the points of alternatives in $\max_{i}(A)$ increase because of the extra points allocated due to the presence of agent $i$. Hence the probability weight of agent $i$'s first equivalence class is strictly more in $MR(\pref)$ in contrast to $MR(\pref_{-i})$.	
		
							 		We now show that $MR$ satisfies strong SD-participation. In view of the our first claim, it will follow that $MR$ satisfies very strong SD-participation. 								We first note that when $i$ votes, he always contributes his score to alternatives in a set $S$ to $\max_{\pref_i}(S)$. When $i$ votes, other agents may also contribute to alternatives in $\max_{\pref_i}(S)$ if they are their most preferred alternatives as well. 
									
							 									Let us consider the recursion tree $T_{-i}$ of $MR$ when $i$ does not vote. Any given node denote by $S$ corresponds to the  set of alternatives $S$ in the recursion tree when $i$ does not vote, for which a total probability weight $v$ needs to be distributed among the children of the node that corresponds to inclusion minimal subsets of $S$ with respect to $(\max_{\pref_1}(S),\ldots, \max_{\pref_n}(S))$.

							 	Let us compare $T_{-i}$ with the recursion tree $T$ of $MR$ when $i$ does vote. Let $S$ be the node nearest to the root in $T_{-i}$ that has a different weight distribution among its children than in $T$. 
							 	Let us compare the children nodes and their probability with the case when $i$ does vote and the children nodes are $S_1', S_2'\ldots$. We show that when $i$ votes, in $S$, any reallocation of probability weight is to $i$'s most preferred alternatives in $S$ which is consistent with an \sd-improving probability transfer for $i$.

							When agents give points to the different alternatives in $S$, then each alternative in $S\setminus \max_{\pref_i}(S)$ get at most the same number of points as when $i$ does not vote. As for alternative in $\max_{\pref_i}(S)$, at least one of them gets strictly more points than before. Therefore when we look at the inclusion minimal subsets in $T$ and $T_{-i}$, any decrease in the probability weight of some $S_j\in S\setminus \max_{\pref_i}(S)$ corresponds to an increase in the probability weight of some $S_{\ell}\subset S \cap \max_{\pref_i}(S)$.

					%
					%
					%
					%
					%

							 	There is a shift in probability weight of the child nodes and each shift in probability weight of the child nodes corresponds to an increase in the probability weight of inclusion minimal subsets that are subsets of $\max_{\pref_i}(S)$ which implies an \sd-improvement. 
							 	The same argument can be applied inductively down the tree which proves that  $MR$ satisfies strong \sd-participation.
							 													\end{proof}

							 		In the shape of $MR$, we have an anonymous, neutral, weak \sd-strategyproof, and ex post efficient SDS that satisfies very strong \sd-participation.
							 	However, $MR$ is not \sd-efficient.  Next, we consider an \sd-efficient rule (\egal) already in the literature~\citep{AzSt14a} and show that it satisfies strong \sd-participation. 

						\egal does not satisfy very strong \sd-participation.

							\begin{theorem}
								\egal does not satisfy very strong \sd-participation.
								\end{theorem}
								\begin{proof}
									Consider the profile:
									\begin{align*}
										1&: a, b&
										2&: a,b&
										3&: b,a&
										4&: b,a
										\end{align*}
										Then $ESR({\pref_{-4}})=\frac{1}{2}a+\frac{1}{2}b$ and $ESR({\pref})=\frac{1}{2}a+\frac{1}{2}b$ as well.
			
												\end{proof}

												In fact,  we next show that \egal does not even satisfy strong \sd-participation.
							
												\begin{theorem}
								\egal does not satisfy strong \sd-participation				\end{theorem}
								\begin{proof}
									Consider the following preference profile.\footnote{The example was identified by Pang Luo using  \url{https://www-m9.ma.tum.de/games/sr-applet/index_en.html}.} 
									\begin{align*}
					         	 1:& \{b,c,f\},\{a,d,e,g,h\}\\
					         	 2:& \{a,h\},\{c,d,e,f,g\},\{b\}\\
					            3:& \{b,c,d,e,h\},\{a,f,g\}\\
					            4:& \{a,d\},\{b,c,g\},\{e\},\{f,h\}\\
					           5: &\{a,d,e,f,h\},\{b,g\},\{c\}\\
					            6:& \{e,h\},\{a,c,f\},\{b,d,g\}
					\end{align*}

					The \egal outcome is following lottery:
			
					$a:0.333333, b:0.166667, c:0.166667, d:0.000000, e:0.000000, f:0.000000, g:0.000000, h:0.333333$.

					The \egal outcome when 2 abstains is
					$a:0.222222, b:0.111111, c:0.222222, d:0.111111, e:0.333333, f:0.000000, g:0.000000, h:0.000000$.

					Note that probability of $b$ is 1/6 in the original lottery but 
					1/9 when 2 abstains. Therefore when 2 votes, the outcomes does not \sd-dominate the outcome when 2 abstains. 
									\end{proof}

            					 						\begin{remark}
            					For a complete refinement $\mathcal{E}$ of \sd,	if a rule is $\mathcal{E}$-efficient, it does not imply that it satisfies very strong $\mathcal{E}$-participation. 							For example, $ESR$ is \dl-efficient but does not satisfy very strong \dl-participation.	
            							\end{remark}

            	%
            	%

            							\citet{BBH15b} proved that RSD satisfies very strong \sd-participation. On the other hand, serial dictatorship that is an \sd-strategyproof and \dl-efficient rule does not even satisfy very strong \dl-participation or very strong \sd-participation.
						
            										\begin{theorem}
            							Serial dictatorship does not satisfy very strong \dl-participation or very strong \sd-participation.
            											\end{theorem}		
            												\begin{proof}
												
            												Consider the following preference profile. 
            												\begin{align*}
            													1&: \{a, b\},c\\
            													2&: c,b,a\\
            													3&: c,b,a
            													\end{align*}
            													Consider serial dictatorship with respect to permutation $123$. Then, the outcome for profile $(\pref_1,\pref_2)$ is $1b$. The outcome remains the same for profile $(\pref_1,\pref_2,\pref_3)$.
            													\end{proof}

            						The theorem above leads to the following observations.
						
            							\begin{remark}
            													If a rule satisfies $\dl$-strategyproofness, it does not imply that it satisfies very strong \dl-participation.					Serial dictatorship is \dl-strategyproof and in fact even \sd-strategyproof. However, we have shown that it does not satisfy very strong \dl-participation.
            								\end{remark}
								
            								\begin{remark}
            														If a rule satisfies $\dl$-efficiency and hence $\sd$-efficiency, it does not imply that it satisfies very strong \dl-participation or very strong \sd-participation. For example, serial dictatorship	satisfies the efficiency properties but not the participation properties.
            									\end{remark}
						
            			%

            			\section{Conclusions}
			
            			In this paper, we continued the line of research concerning strategic aspects in randomized social choice~\citep{Aziz13b,ABB13d, ABBH12a,BBH15b}. In particular, we expanded the taxonomy of participation notions by relating participation with respect to stochastic dominance and participation with respect to refinements of stochastic dominance. We also proved  \PS satisfies very strong-participation but \egal does not satisfy strong \sd-participation. 
				Other than identifying attractive SDSs with good efficiency and participation properties, our study further clarifies the the extent to which efficiency, polynomial-time computability, and participation are compatible.  
			
            			\citet{BBH15b} posed an interesting open problem whether there exists a social decision scheme that satisfies \sd-efficiency and very strong \sd-efficiency. We conclude with a similar open problem: does there exist an SDS that satisfies \dl-efficiency and very strong \dl-participation?

                        \section*{Acknowledgments}

		   The author thanks Felix Brandt, Pang Luo and Christine Rizkallah for discussion on the topic. The author also thanks Nicholas Mattei for useful feedback.

\begin{thebibliography}{20}
\expandafter\ifx\csname natexlab\endcsname\relax\def\natexlab#1{#1}\fi
\expandafter\ifx\csname url\endcsname\relax
  \def\url#1{\texttt{#1}}\fi
\expandafter\ifx\csname urlprefix\endcsname\relax\def\urlprefix{URL }\fi

\bibitem[{Abdulkadiro\u{g}lu and S{\"o}nmez(2003)}]{AbSo03a}
Abdulkadiro\u{g}lu, A., S{\"o}nmez, T., 2003. Ordinal efficiency and dominated
  sets of assignments. Journal of Economic Theory 112~(1), 157--172.

\bibitem[{Aziz(2013)}]{Aziz13b}
Aziz, H., 2013. {Maximal Recursive Rule: A New Social Decision Scheme}. In:
  Proc.~of 22nd IJCAI. AAAI Press, pp. 34--40.

\bibitem[{Aziz et~al.(2014)Aziz, Brandl, and Brandt}]{ABB13d}
Aziz, H., Brandl, F., Brandt, F., 2014. On the incompatibility of efficiency
  and strategyproofness in randomized social choice. In: Proc.~of 28th AAAI
  Conference. AAAI Press, pp. 545--551.

\bibitem[{Aziz et~al.(2015)Aziz, Brandl, and Brandt}]{ABB14b}
Aziz, H., Brandl, F., Brandt, F., 2015. Universal dominance and welfare for
  plausible utility functions. Journal of Mathematical Economics 60, 123--133.

\bibitem[{Aziz et~al.(2013{\natexlab{a}})Aziz, Brandt, and Brill}]{ABB13b}
Aziz, H., Brandt, F., Brill, M., 2013{\natexlab{a}}. The computational
  complexity of random serial dictatorship. Economics Letters 121~(3),
  341--345.

\bibitem[{Aziz et~al.(2013{\natexlab{b}})Aziz, Brandt, and Brill}]{ABBH12a}
Aziz, H., Brandt, F., Brill, M., 2013{\natexlab{b}}. On the tradeoff between
  economic efficiency and strategyproofness in randomized social choice. In:
  Proc.~of 12th AAMAS Conference. IFAAMAS, pp. 455--462.

\bibitem[{Aziz and Stursberg(2014)}]{AzSt14a}
Aziz, H., Stursberg, P., 2014. A generalization of probabilistic serial to
  randomized social choice. In: Proc.~of 28th AAAI Conference. AAAI Press, pp.
  559--565.

\bibitem[{Bogomolnaia(2015)}]{Bogo15a}
Bogomolnaia, A., 2015. The most ordinally-egalitarian of random voting rules.
  In: The proceedings of the 3rd International Workshop on Matching Under
  Preferences (MATCHUP).

\bibitem[{Bogomolnaia and Moulin(2001)}]{BoMo01a}
Bogomolnaia, A., Moulin, H., 2001. A new solution to the random assignment
  problem. Journal of Economic Theory 100~(2), 295--328.

\bibitem[{Bogomolnaia et~al.(2005)Bogomolnaia, Moulin, and Stong}]{BMS05a}
Bogomolnaia, A., Moulin, H., Stong, R., 2005. Collective choice under
  dichotomous preferences. Journal of Economic Theory 122~(2), 165--184.

\bibitem[{Brandl et~al.(2015{\natexlab{a}})Brandl, Brandt, and
  Hofbauer}]{BBH15b}
Brandl, F., Brandt, F., Hofbauer, J., 2015{\natexlab{a}}. Incentives for
  participation and abstention in probabilistic social choice. In: Proc.~of
  14th AAMAS Conference. IFAAMAS, pp. 1411--1419.

\bibitem[{Brandl et~al.(2015{\natexlab{b}})Brandl, Brandt, and
  Hofbauer}]{BBH15d}
Brandl, F., Brandt, F., Hofbauer, J., 2015{\natexlab{b}}. Welfare maximization
  entices participation. Tech. rep., http://arxiv.org/abs/1508.03538.

\bibitem[{Cho(2012)}]{Cho12a}
Cho, W.~J., 2012. Probabilistic assignment: A two-fold axiomatic approach,
  {M}imeo.

\bibitem[{Fishburn and Brams(1983)}]{BrFi83a}
Fishburn, P.~C., Brams, S.~J., 1983. Paradoxes of preferential voting.
  Mathematics Magazine 56~(4), 207--214.

\bibitem[{Gibbard(1977)}]{Gibb77a}
Gibbard, A., 1977. Manipulation of schemes that mix voting with chance.
  Econometrica 45~(3), 665--681.

\bibitem[{Moulin(1988)}]{Moul88b}
Moulin, H., 1988. Condorcet's principle implies the no show paradox. Journal of
  Economic Theory 45, 53--64.

\bibitem[{Moulin(2003)}]{Moul03a}
Moulin, H., 2003. Fair Division and Collective Welfare. The MIT Press.

\bibitem[{Procaccia(2010)}]{Proc10a}
Procaccia, A., 2010. Can approximation circumvent {G}ibbard-{S}atterthwaite?
  In: Proc.~of 24th AAAI Conference. AAAI Press, pp. 836--841.

\bibitem[{Saban and Sethuraman(2014)}]{SaSe13b}
Saban, D., Sethuraman, J., 2014. A note on object allocation under
  lexicographic preferences. Journal of Mathematical Economics 50, 283--289.

\bibitem[{Schulman and Vazirani(2012)}]{ScVa12a}
Schulman, L.~J., Vazirani, V.~V., 2012. Allocation of divisible goods under
  lexicographic preferences. Tech. Rep. arXiv:1206.4366, arXiv.org.

\end{thebibliography}

		\end{document}